\documentclass[journal]{IEEEtran}

\usepackage{epsf}
\usepackage{graphicx}
\usepackage{amsfonts,amssymb,amsmath}
\usepackage{algorithmic,algorithm}
\usepackage{subfigure}
\usepackage{latexsym}
\usepackage{bm}
\usepackage[usenames]{color}
\usepackage{url}

\usepackage[normalem]{ulem}

\newcommand{\expect}[1]{\mathbb{E}\left[#1\right]}

\newtheorem{theorem}{Theorem}
\newtheorem{lemma}{Lemma}
\newtheorem{assumption}{Assumption}

\begin{document}
\title{Delay Optimal Power Aware Opportunistic Scheduling with Mutual Information Accumulation}
\author{Xiaohan Wei and Michael J. Neely
\thanks{The authors are with the Electrical Engineering department at the University of Southern California, Los Angeles, CA.}
\thanks{This material is supported in part by the NSF grant 0964479.}}
\maketitle

\begin{abstract}
This paper considers optimization of power and delay in a time-varying
wireless link using rateless codes.  The link serves a sequence of variable-length
packets.  Each packet is coded and transmitted over multiple slots. Channel
conditions can change from slot to slot and are unknown to the transmitter.  The
amount of mutual information accumulated on each slot depends on the random
channel realization and the power used. The goal is to minimize average service
delay subject to an average power constraint. We formulate this problem as a frame-based
stochastic optimization problem and solve it via an online algorithm. We show that the
subproblem within each frame is a simple
integer program which can be effectively solved using a dynamic program. The
optimality of this online algorithm is proved using the frame-based Lyapunov
drift analysis.
\end{abstract}
\section{Introduction}
Consider a slotted time system where a single source transmits a sequence of data packets to a receiver over a time varying channel.  During each time slot, the source transmits a portion of the mutual information needed for the receiver to decode the current packet.  This is done by choosing the amount of power on each transmission.  The source does not know the current channel state when it makes this decision. The chosen power and the random channel state determine the amount of mutual information that is received. Transmission of that packet is finished when the receiver accumulates enough mutual information to decode. This technique is usually referred to as mutual information accumulation (\cite{mutual-info-1}-\cite{mutual-info-2}) and can be implemented through rateless coding schemes such as Raptor codes and Fountain codes (\cite{Fountain-codes}-\cite{Raptor-codes}).

Our goal is to design a transmission algorithm such that the time average delay of the sequence of packets is minimized and a time average power constraint is satisfied. Specifically, we consider an online setting where both the channel conditions and the packet lengths could vary according to some underlying probability distributions. This problem is challenging for two reasons: (i) Decoding times can vary randomly for each packet, (ii) The power spent during the current time slot affects the desired power constraint and also influences how much information is left to transmit for that packet in  future slots, creating complex dependencies over different slots. Fortunately, the system contains an underlying Markov renewal structure that is exploited in this paper.

In this paper, we use renewal reward theory (e.g. \cite{Gallager-renewal}) together with Lyapunov optimization (e.g. \cite{Neely2010}) to develop an online algorithm for this problem. Previous work \cite{renewal-opt-tac} introduces a general framework for applying Lyapunov optimization  to renewal systems through the idea of dividing the time line into renewal frames and solving a subproblem within each renewal frame. Works
in \cite{Wiopt-Xiaohan} and \cite{MG1-queue} analyze the performance of frame-based Lyapunov optimization for a file downloading problem and a multiclass M/G/1 queue problem, respectively. In both papers, the subproblem within each frame is effectively solved, which is the key to the success of the whole algorithm. In this paper, we show that the subproblem within each frame is an integer program which can be effectively solved using a dynamic program.

Prior work on mutual information accumulation in \cite{broadcast04}\cite{broadcast04-2}\cite{mutual-info-2}\cite{routing11} focus on a network with fixed channel conditions. Works in \cite{broadcast04} and \cite{broadcast04-2} consider the problem of minimum energy accumulative broadcast and show that it is NP-complete. Works \cite{mutual-info-2} and \cite{routing11} consider the problem of optimal cooperative routing in a multi-hop network and develop linear program based formulations. Applications to time-varying channel conditions  are considered in \cite{fading_1} and \cite{fading_2}. All these previous works only consider the optimal transmission of a single packet. Our work differs from these works in that we consider a sequential transmission scenario with a time average power constraint.

The contributions of the paper are as follows: First, we formulate this power constrained minimum delay problem as a frame-based stochastic optimization problem, where each frame is the time period for transmitting one single packet. Second, we propose a provably optimal frame-based online algorithm where an integer program is solved during each frame. We show that the integer program can be effectively solved through dynamic programming.  A nice feature of the resulting algorithm is that, while it assumes knowledge of the channel probability distribution, it does not require knowledge of the file size probability distribution. Further, the power constraint is shown to hold even if there is a mismatch between the actual channel probability distribution and the assumed distribution.

\section{Problem Formulation}
In this section, we present our system model and the delay minimization problem. Consider a slotted time single hop network where one source is transmitting a sequence of packets to a receiver with time varying channel states. The transmitter uses a rateless coding scheme where the receiver accumulates mutual information about each packet until it can successfully decode. The goal is to minimize the time average packet delay subject to an average power constraint on the source. For simplicity of analysis, we assume that the source's only significant power expenditure lies in transmissions, and coding/decoding yields negligible power and time consumption.

Let $\delta>0$ be a fixed amount of mutual information, called a \emph{data unit}.
Let $L[f]$ be the length of the $f$-th packet, which is a positive integer number of data units associated with packet $f$. We assume that $L[f]$ is i.i.d. over $f\in\{0,1,2,\cdots\}$ and takes values in a finite set $\mathcal{L}$.
Also, assume that the source uses \emph{ideal rateless codes} for transmission so that the mutual information collected by the receiver from different slots for the same packet add up with each other. Then, the receiver can decode packet $f$ if and only if it accumulates at least $L[f]$ data units for that packet. Finally, we assume that the source cannot transmit any new packet until the transmission of the previous  packet is finished.

The source has different power spectral density (PSD) options to take, which is in units of joules/slot/Hz. Let $P(t)\in\mathcal{P}$ be the option taken by the source at slot $t$. We assume $\mathcal{P}=\{P_1,P_2,\cdots,P_{|\mathcal{P}|}\}$ is finite. Without loss of generality, further assume  $0<P_1< P_2<\cdots< P_{|\mathcal{P}|}$. The channel power gain at time slot $t$ is denoted as $\alpha(t)$. We assume $\alpha(t)>0$ is i.i.d. over slot and takes values in a finite set $\mathcal{A}=\{\alpha_1,\alpha_2,\cdots,\alpha_{|\mathcal{A}|}\}$ with probability distribution $\phi(\alpha_i)$. They satisfy
\[\phi(\alpha_i)\geq0,~\sum_{i=1}^{|\mathcal{A}|}\phi(\alpha_i)=1.\]
Then, the mutual information transmitted during a single time slot $t$ can be expressed as $K(\alpha(t),P(t))$, which is also a positive integer number of data units. We make no assumption on the specific form of $K(\alpha(t),P(t))$ except that it must be nondecreasing with respect to both $\alpha(t)$ and $P(t)$. For example, if the source uses fountain codes, then by Shannon's formula, $K(\alpha(t),P(t))=\left\lceil\log_2\left(1+\frac{\alpha(t)P(t)}{N_0}\right)\right\rceil$, where $N_0/2$ denotes the PSD of the white noise process.

%\footnote{
%If the source uses ideal fountain codes, then by Shannon's formula, $f_W(\alpha(t),P(t))$ is
%\[\log_2\left(1+\frac{\alpha(t)P(t)W}{N_0W}\right)=\log_2\left(1+\frac{\alpha(t)P(t)}{N_0}\right),\]
%where $N_0/2$ denotes the PSD of the white noise process.}.

For the ease of algorithm development, we further segment the slotted timeline into \emph{frames}. Define frame $f$ as the time period (in units of slots) of transmitting packet $f$, with starting slot denoted as $t_f$ and frame length denoted as $T[f]\triangleq t_{f+1}-t_f$. Frame $f$ ends on the first slot in which the total mutual information for that frame exceeds $L[f]$, which is
\[\sum_{t=t_f}^{t_{f+1}-1}K(\alpha(t),P(t))\geq L[f].\]
We make the following assumption regarding the random processes $\{L[f]\}_{f=0}^\infty$ and $\{\alpha(t)\}_{t=0}^\infty$:
\begin{assumption}
The distribution of packet length process $\{L[f]\}_{f=0}^\infty$ is unknown to the source but the source can observe $L[f]$ at the beginning of each frame. The process $\{\alpha(t)\}_{t=0}^\infty$ is \emph{unobservable}, meaning that the source cannot see $\alpha(t)$ before transmission at slot $t$, but the source has the knowledge of distribution $\phi(\alpha_i),~\alpha_i\in\mathcal{A}$.
\end{assumption}

Furthermore, at the end of slot $t$, the transmitter provides a feedback signal so that the source effectively knows $\alpha(t)$ and thus how much mutual information has been transmitted.
The goal is to choose the PSD allocation $P(t)$ at each time slot so as to minimize the average delay (in units of slots/packet) subject to a time average power constraint:
\begin{align*}
\min~~&\lim_{F\rightarrow\infty}\frac{\sum_{f=0}^{F-1}T[f]}{F}\\
\textrm{s.t.}~~&\lim_{F\rightarrow\infty}\frac{\sum_{f=0}^{F-1}\sum_{t=t_f}^{t_{f+1}-1}P(t)}{\sum_{f=0}^{F-1}T[f]}\leq\beta,
\end{align*}
where $\beta$ is a positive constant. An equivalent formulation which does not include a ratio of time averages in the constraint is as follows:
\begin{align}
\min~~&\lim_{F\rightarrow\infty}\frac{1}{F}\sum_{f=0}^{F-1}T[f]\label{eq_1}\\
s.t.~~&\lim_{F\rightarrow\infty}\frac{1}{F}\sum_{f=0}^{F-1}\sum_{t=t_f}^{t_{f+1}-1}\left(P(t)-\beta\right)\leq0 \label{eq_2}.
\end{align}
For the rest of the paper, we shall consider \eqref{eq_1}-\eqref{eq_2} instead.
We also need the assumption that the minimal PSD allocation $P_1$ is greater than 0 but small enough. This ensures that the frame length $T[f]$ is finite for any policy, and the problem described is feasible.
\begin{assumption}\label{assumption_feasible}
The minimal PSD option $P_1$ satisfies $0<P_1<\beta$.
\end{assumption}

Furthermore, let
\begin{align*}
  L_{\max}&=\max_{L\in\mathcal{L}}~L,\\
  K_{\min}&=\min_{\alpha\in\mathcal{A}}~K(\alpha, P_1).
\end{align*}
Here $L_{\max}$ stands for maximum packet length and $K_{\min}$ stands for minimum per slot mutual information transmission. Notice that since $K(\alpha, P)$ is a positive integer for any $\alpha$ and $P$, we have $K_{\min}\geq1$.

\section{An Online Algorithm}
In this section, we propose an online algorithm which adapts to the time varying packet lengths and channel power gains.
\subsection{Proposed algorithm}\label{online_algorithm}
The proposed algorithm operates over frames, and features a ``virtual queue'' technique treating the time average power constraint. Define a virtual queue $Q[f]$, which has initial condition $Q[0]=0$ and update equation:
\begin{equation}\label{queue_update}
Q[f+1]=\max\left\{Q[f]+\sum_{t=t_f}^{t_{f+1}-1}(P(t)-\beta),~0\right\}.
\end{equation}
The proposed algorithm then runs in the following two steps with a trade-off parameter $V>0$:
\begin{itemize}
  \item At the beginning of frame $f$, observe $Q[f]$, $L[f]$, and choose a sequence of PSD allocations so as to solve the following optimization problem:
      \begin{align}
      \min_{P(t)\in\mathcal{P}}&~~\expect{\left.
              \sum_{t=t_f}^{t_{f+1}-1}R_f(P(t))\right| Q[f], L[f]}\label{em_aver_1}\\
      \textrm{s.t.}&~~\sum_{t=t_f}^{t_{f+1}-1}K(\alpha(t), P(t))\geq L[f]\label{em_aver_2}
      \end{align}
      where for any allocation $P(t)\in\mathcal{P}$, the function $R_f(P(t))$ is defined as $R_f(P(t))\triangleq V+Q[f](P(t)-\beta)$, and
      the expectation is taken with respect to the random channel gain $\alpha(t)$.
  %\item Update $\hat{\theta}[f]$ as follows:
%  \begin{equation}\label{hat_theta_update}
%  \hat{\theta}[f+1]=\frac{1}{(f+1)^\xi}\sum_{i=0}^f\left(T[i]-\theta[i]+\frac{Q[i]}{V}\sum_{t=t_i}^{t_{i+1}-1}\left(P(t)-\beta\right)\right)
%  \end{equation}
%  \item Update $\theta[f]$ as follows:
%  \[\theta[f+1]=\left\{
%                  \begin{array}{ll}
%                    \hat{\theta}[f+1], & \hbox{if $\hat{\theta}[f+1]\in (1,\lceil L_{\max}/K_{\min}\rceil)$;} \\
%                    \lceil L_{\max}/K_{\min}\rceil, & \hbox{if $\hat{\theta}[f+1]\geq \lceil L_{\max}/K_{\min}\rceil$;} \\
%                    1, & \hbox{if $\hat{\theta}[f+1]\leq 1$.}
%                  \end{array}
%                \right.
%  \]
%
  \item Update $Q[f]$ according to \eqref{queue_update} at the end of each frame.
\end{itemize}
Section \ref{section_perform_analysis} shows the above algorithm structure
produces a near optimal solution, where proximity to optimality is determined
by the $V$ parameter.  The following section shows how to implement each
state of the algorithm, that is, how to solve the problem \eqref{em_aver_1}-\eqref{em_aver_2}.

\subsection{Solve \eqref{em_aver_1}-\eqref{em_aver_2} via dynamic program under static channel gain}
In this section, we show how to solve \eqref{em_aver_1}-\eqref{em_aver_2} during each frame under static channel gain, i.e. $\alpha(t)=\alpha_0,~\forall t$, where $\alpha_0$ is a positive constant.
Then, we show how to generalize the algorithm to random channel case in the next section.

Given $Q[f]$ and $L[f]$, there is no randomness in the formulation \eqref{em_aver_1}-\eqref{em_aver_2}. Thus, we can simplify it as follows:
\begin{align}
      \min_{P(t)\in\mathcal{P}}&~~\sum_{t=t_f}^{t_{f+1}-1}R_f(P(t))\label{sub_dp}\\
      \textrm{s.t.}&~~\sum_{t=t_f}^{t_{f+1}-1}K(\alpha_0, P(t))\geq L[f]. \label{sub_dp_2}
\end{align}
Next, we distinguish between the following two cases:
\begin{enumerate}
  \item  Case 1: Suppose there exists some option $P_{j_0}\in\mathcal{P}$ such that $R_f(P_{j_0})<0$, then, since both $R_f(P(t))$ and $K(\alpha_0,P(t))$ are increasing with $P(t)$, the best decision is always choosing the smallest power allocation (i.e. $P(t)=P_1$) till the end of the frame.
  \item Case 2: Suppose $R_f(P_j)\geq0,~\forall P_j\in\mathcal{P}$. Then, we can interpret above problem in the following way: For each PSD option $P_j\in\mathcal{P}$, there is a corresponding penalty $R_f(P_j)$ and a weight $K(\alpha_0, P_j)$. The goal is to minimize the total penalty subject to the constraint that the total weight should be at least $L[f]$. Since there are only finite options in $\mathcal{P}$, \eqref{sub_dp}-\eqref{sub_dp_2} can be written as the following integer program problem:
      \begin{align*}
      \min_{x_j,~j\in\{1,\cdots,|\mathcal{P}|\}}~~&\sum_{j=1}^{|\mathcal{P}|}R_f(P_j)x_j\\
      \textrm{s.t.}~~&\sum_{j=1}^{|\mathcal{P}|}K(\alpha_0, P_j)x_j\geq L[f],~x_j\in\mathbb{N}.
      \end{align*}
      Notice that both $K(\alpha_0, P_j)$ and $L[f]$ are positive integers. This is similar to a classical Knapsack problem. A standard technique solving this type of problems is dynamic program (see \cite{integer-programming} for related results).
      Define the state variable $k,~k\in\mathbb{Z}$ and define $m[k],~k\in\mathbb{Z}$ as the minimum value achieved by solving the following problem:
\begin{align}
\min_{x_j,~j\in\{1,\cdots,|\mathcal{P}|\}}~~&\sum_{i=1}^{|\mathcal{P}|}R_f(P_j)x_j   \label{int_program_1}\\
s.t.~~&\sum_{i=1}^{|\mathcal{P}|}K(\alpha_0,P_j)x_j\geq k,~x_j\in\mathbb{N}. \label{int_program_2}
\end{align}
Define $m[k]=0,~\forall k\leq0$, then, we progessively compute $m[k]$ for each $k\in\{1,2,\cdots,L[f]\}$ as follows:
\begin{align}
m[k]&=\min_{j\in\{1,2,\cdots,|\mathcal{P}|\}}~R_f(P_j)+m[k-K(\alpha_0, P_j)],\nonumber\\
&~\forall k>0.   \nonumber
\end{align}
The proof of optimality is similar to that of knapsack problem which can be found, for example, in chapter 6 of \cite{integer-programming}.

\end{enumerate}

\subsection{Extension to random channel gain}\label{dp_random_channel}
In this section, we solve \eqref{em_aver_1}-\eqref{em_aver_2} under random channel gains. Again, we distinguish between the following two cases:
\begin{enumerate}
  \item Suppose there exists some option $P_{j_0}\in\mathcal{P}$ such that $R_f(P_{j_0})<0$. Then, since both $R_f(P(t))$ and $K(\alpha(t),P(t))$ are increasing with $P(t)$ for any $\alpha(t)\in\mathcal{A}$, the best decision is always choosing the smallest power allocation (i.e. $P(t)=P_1$) until the end of the frame.
  \item Suppose $R_f(P_j)\geq0,~\forall P_j\in\mathcal{P}$. Then, in a similar way as before, we get a randomized integer program, which can still be solved by dynamic program.
  Define $k,~k\in\mathbb{Z}$ as the state variable and define $m[k]$, $k\in\mathbb{Z}$ as the minimum value achieved by the following problem:
       \begin{align}
       \min_{P(t)\in\mathcal{P}}&~~\expect{\left.VT[f]
              +Q[f]\sum_{t=t_f}^{t_{f+1}-1}\left(P(t)-\beta\right)\right| Q[f], L[f]}\label{appr_aver_1}\\
       s.t.&~~\sum_{t=t_f}^{t_{f+1}-1}K(P(t),\alpha(t))\geq k\label{appr_aver_2}
      \end{align}
   Define $m[k]=0,~\forall k\leq0$, then, we progessively compute $m[k]$ for each $k\in\{1,2,\cdots,L[f]\}$ as follows
   \begin{align}
       m[k]=\min_{j\in\{1,2,\cdots,|\mathcal{P}|\}}&~R_f(P_j)+\sum_{i=1}^{|\mathcal{A}|}\phi(\alpha_i)m[k\nonumber\\
       &-K(\alpha_i,P_j)],~\forall k>0, \label{value_function_2}
   \end{align}
   The proof that \eqref{value_function_2} indeed leads to optimal policy during each frame is similar to that of a standard proof of optimality for a dynamic program (see \cite{DP&control} for more details). For completeness, we provide a proof using induction in Appendix \ref{proof_dp}.
\end{enumerate}

\section{Performance Analysis}\label{section_perform_analysis}
In this section, we prove that the online algorithm in Section \ref{online_algorithm} satisfies the power constraint and achieves the near optimal throughput (with $\mathcal{O}(1/V)$ optimality gap). We first show that
the virtual queue $Q[f]$ is bounded, which implies that the time average of ``arrival process'' $\sum_{t=t_f}^{t_{f+1}-1}P(t)$ is less than or equal to that of ``service process'' $\beta T[f]$, and thus the power constraint is satisfied. Then, we prove that the algorithm gives near optimal performance by Lyapunov drift analysis.

\subsection{Average power constraint via queue bound}\label{section_constraint}
\begin{lemma}
If there exists a constant $C>0$ such that $Q[f]\leq C$, then, it follows for any $F>0$,
\[\frac{1}{F}\sum_{f=0}^{F-1}\sum_{t=t_f}^{t_{f+1}-1}\left(P(t)-\beta\right)\leq\frac{C}{F},\]
thus,
\[\limsup_{F\rightarrow\infty}\frac{1}{F}\sum_{f=0}^{F-1}\sum_{t=t_f}^{t_{f+1}-1}\left(P(t)-\beta\right)\leq0.\]
\end{lemma}
\begin{proof}
From \eqref{queue_update}, we know that for each frame $f$:
\begin{equation}
    Q[f+1] \geq Q[f]
     +\sum_{t=t_f}^{t_{f+1}-1}(P(t)-\beta) \nonumber
\end{equation}
Rearranging terms gives:
\[ \sum_{t=t_f}^{t_{f+1}-1}(P(t)-\beta) \leq Q[f+1]-Q[f] \]
Fix $F>0$. Summing over $f\in\{0,1,\cdots,F-1\}$  gives:
\begin{eqnarray*}
    \sum_{f=0}^{F-1}\sum_{t=t_f}^{t_{f+1}-1}(P(t)-\beta) \leq Q[F] - Q[0] \leq C,
\end{eqnarray*}
which implies:
\[ \frac{1}{F}\sum_{f=0}^{F-1}\sum_{t=t_f}^{t_{f+1}-1}(P(t)-\beta) \leq \frac{C}{F}. \]
Taking $F\rightarrow\infty$ gives the desired result.
\end{proof}

The next lemma shows that the queue process deduced from proposed algorithm is deterministically bounded.

\begin{lemma}
If $Q[0]=0$, then, the proposed algorithm gives for all $f\in\mathbb{N}^+$,
\[Q[f]\leq\max\left\{\frac{V}{\beta-P_1}+\left\lceil\frac{L_{\max}}{K_{\min}}\right\rceil\cdot(P_{|\mathcal{P}|}-\beta),~0\right\},\]
where $\lceil x\rceil$ denotes the ceiling function rounding up towards the smallest integer greater than $x$.
\end{lemma}
\begin{proof}
First, consider the case when $P_{|\mathcal{P}|}\leq\beta$. From \eqref{queue_update} it is obvious that $Q[f]$ can never increase. Thus, $Q[f]=0,~\forall f\in\mathbb{N}^+$.

Next, consider the case when $P_{|\mathcal{P}|}>\beta$.
It is obvious that the frame length $T[f]$ is upper bounded by $\left\lceil\frac{L_{\max}}{K_{\min}}\right\rceil$. We then prove the assertion by induction on $f$. The result trivially holds for $f=0$.  Suppose it holds at $f=l$ for $l\in\mathbb{N}^+$, so that:
\[Q[l]\leq\frac{V}{\beta-P_1}+\left\lceil\frac{L_{\max}}{K_{\min}}\right\rceil\cdot(P_{|\mathcal{P}|}-\beta).\]
We are going to prove the same holds for $f=l+1$. There are two further cases to be considered:
\begin{enumerate}
  \item $Q[l]\leq\frac{V}{\beta-P_1}$. In this case we have by \eqref{queue_update}:
  \begin{eqnarray*}
   Q[l+1] &\leq& Q[l] + T[l](P_{|\mathcal{P}|} - \beta) \\
   &\leq& \frac{V}{\beta-P_1} + \left\lceil\frac{L_{\max}}{K_{\min}}\right\rceil\cdot(P_{|\mathcal{P}|} - \beta).
   \end{eqnarray*}

    \item $\frac{V}{\beta-P_1}<Q[l]\leq\frac{V}{\beta-P_1} + \left\lceil\frac{L_{\max}}{K_{\min}}\right\rceil\cdot(P_{|\mathcal{P}|} - \beta)$. In this case, we have by first half of induction hypothesis,
        \[Q[l]P_1+V-\beta Q[l]<0.\]
        By the first case in Section \ref{dp_random_channel}, the source will choose power allocation $P(t)=P_1,~\forall t\in[t_l, t_{l+1}-1]$. Since $P_1<\beta$ by assumption \ref{assumption_feasible}, it follows from the queue updating rule \eqref{queue_update} that the queue will increase no more at the end of frame $l$. By the second half of the induction hypothesis,
        \[Q[l+1]\leq\frac{V}{\beta-P_1} + \left\lceil\frac{L_{\max}}{K_{\min}}\right\rceil\cdot(P_{|\mathcal{P}|} - \beta).\]
\end{enumerate}
This finishes the proof.
\end{proof}
Combining above two lemmas gives the result that the proposed algorithm satisfies the power constraint. Notice that above result is a sample path result which only assumes $K_{\min}>0$, $\beta-P_1>0$ and probability $\phi(\cdot)\in[0,1]$. Thus, the algorithm meets the average power constraints even if it uses incorrect values for these parameters.

\subsection{Optimization over randomized stationary algorithm}\label{randomized_stationary}
Suppose the server knows the probability distribution of packet length. Then, consider the following class of algorithm: At the beginning of each frame $f$, after observing the packet length $L[f]$, the source independently selects a sequence of power allocation $P^*(t), t\in[t_{f},~t_{f+1}-1]$ according to some pre-specified probability distribution which depends only on $L[f]$ and feedback during that frame.

Then, since $L[f]$ is an i.i.d. process over frames, this algorithm is stationary. Furthermore, let $\{T^*[f]\}_{f=0}^\infty$ and $\left\{\sum_{t=t_f}^{t_{f}+T^*[f]-1}P^*(t)\right\}_{f=0}^\infty$ be two processes corresponding to this randomized algorithm. It can be shown that these two processes are both i.i.d. over frames. By strong law of large numbers,
\begin{align*}
&\lim_{F\rightarrow\infty}\frac{1}{F}\sum_{f=0}^{F-1}T^*[f]
=\expect{T^*[f]},~w.p.1,\\
&\lim_{F\rightarrow\infty}\frac{1}{F}\sum_{f=0}^{F-1}\sum_{t=t_f}^{t_{f}+T^*[f]-1}\left(P^*(t)-\beta\right)\nonumber\\
&=\expect{\sum_{t=t_f}^{t_{f}+T^*[f]-1}\left(P^*(t)-\beta\right)},~w.p.1.
\end{align*}
Furthermore, it can be shown that the optimal solution to \eqref{eq_1}-\eqref{eq_2} can be achieved over this class of algorithms (see \cite{renewal-opt-tac} for related results). That being said, there exists a randomized stationary algorithm such that
\begin{align}
&\expect{T^*[f]}=\theta^*,\label{iid_theta}\\
&\expect{\sum_{t=t_f}^{t_{f}+T^*[f]-1}\left(P^*(t)-\beta\right)}\leq0,\label{iid_P}
\end{align}
where $\theta^*$ is the optimal delay solving \eqref{eq_1}-\eqref{eq_2}. As we shall see, this result plays an important role in our subsequent proof of optimality. Unfortunately, this algorithm is not implementable because the source has no knowledge on the distribution of packet length $L[f]$. Even if it has the knowledge, the computation complexity is very high because of the possibly large packet length set $\mathcal{L}$.

\subsection{Key feature inequality}
We give a feature of our proposed online algorithm which is the key to the proof of optimality.
Define $\mathcal{H}[f]$ as the \emph{system history} up to frame $f$. Consider the algorithm that at the beginning of frame $f$, observes $Q[f]$, $L[f]$ and then chooses a sequence of PSD allocations $P(t)\in\mathcal{P}$ so as to minimize \eqref{em_aver_1}-\eqref{em_aver_2}. Then, the policy must have achieved a smaller value on \eqref{em_aver_1} compared to that of best stationary policy from Section \ref{randomized_stationary}. Mathematically, this can be expressed as the following inequality:
\begin{align*}
&\expect{\left.VT[f]+Q[f]\sum_{t=t_f}^{t_{f+1}-1}\left(P(t)-\beta\right)\right| \mathcal{H}[f],L[f]}\\
\leq&\expect{\left.VT^*[f]+Q[f]\sum_{t=t_f}^{t_{f}+T^*[f]-1}\left(P^*(t)-\beta\right)\right| \mathcal{H}[f],L[f]},
\end{align*}
where the left-hand-side uses values $T[f]$ and $P(t)$ from the implemented policy, and the right-hand-side uses values $T^*[f]$ and $P^*(t)$ from the stationary policy.

Taking expectation from both sides regarding $L[f]$ gives
\begin{align*}
&\expect{\left.VT[f]+Q[f]\sum_{t=t_f}^{t_{f+1}-1}\left(P(t)-\beta\right)\right| \mathcal{H}[f]}\\
\leq&\expect{\left.VT^*[f]+Q[f]\sum_{t=t_f}^{t_{f}+T^*[f]-1}\left(P^*(t)-\beta\right)\right| \mathcal{H}[f]}.
\end{align*}
Using the fact that the stationary randomized policy is independent of history as well as \eqref{iid_theta} and \eqref{iid_P},
\begin{equation}\label{key_feature}
\expect{\left.VT[f]+Q[f]\sum_{t=t_f}^{t_{f+1}-1}\left(P(t)-\beta\right)\right| \mathcal{H}[f]}\leq V\theta^*
\end{equation}

\subsection{Near optimal performance}
\begin{theorem}
The proposed algorithm satisfies the constraint and yields the time average delay satisfying
\[\limsup_{F\rightarrow\infty}\frac{1}{F}\sum_{f=0}^{F-1}T[f]\leq\theta^*+\frac{C_0}{V},\]
with probability 1, where $C_0$ is a positive constant independent of $V$.
\end{theorem}
\begin{proof}
First, it is obvious from the two lemmas in Section \ref{section_constraint} that the time average constraint is satisfied. The rest is devoted to proving the delay guarantee.

Define the Lyapunov drift $\Delta[f]$ for the virtual queue $Q[f]$ as follows
\[\Delta[f]=\frac{1}{2}\left(Q[f+1]^2-Q[f]^2\right).\]
Substitute the definition of $Q[f+1]$ gives
\begin{align*}
\Delta[f]\leq&\frac{1}{2}\left(Q[f]+\sum_{t=t_f}^{t_{f+1}-1}(P(t)-\beta)\right)^2-\frac{1}{2}Q[f]^2\\
=&\frac{1}{2}\left(\sum_{t=t_f}^{t_{f+1}-1}(P(t)-\beta)\right)^2+Q[f]\sum_{t=t_f}^{t_{f+1}-1}(P(t)-\beta).
\end{align*}
Recall that the frame length $T[f]$ is upper bounded by $\left\lceil\frac{L_{\max}}{K_{\min}}\right\rceil$ and $P(t)\leq P_{|\mathcal{P}|}$. Let $C_0\triangleq\frac{1}{2}\left(\left\lceil\frac{L_{\max}}{K_{\min}}\right\rceil(P_{|\mathcal{P}|}+\beta)\right)^2$, it follows
\[\Delta[f]\leq C_0+Q[f]\sum_{t=t_f}^{t_{f+1}-1}(P(t)-\beta).\]
Now, we consider the following expected drift term $\Delta[f]$ plus penalty term $VT[f]$:
\begin{align*}
&\expect{\Delta[f]+VT[f]|\mathcal{H}[f]}\\
\leq& C_0+\expect{\left.Q[f]\sum_{t=t_f}^{t_{f+1}-1}(P(t)-\beta)+VT[f]\right|\mathcal{H}[f]}\\
\leq& C_0+V\theta^*,
\end{align*}
where the last inequality follows from \eqref{key_feature}. Thus, rearranging the terms gives
\begin{equation}
\expect{\Delta[f]+V(T[f]-\theta^*)|\mathcal{H}[f]}\leq C_0.\label{intermediate_1}
\end{equation}
Since $Q[f]$ is deterministically bounded, $\Delta[f]$ is also deterministically bounded, thus, the following holds
\begin{equation}
\sum_{f=0}^{\infty}\frac{\expect{\Delta[f]^2}}{f^2}<\infty.\label{intermediate_2}
\end{equation}
Corollary 4.2 of \cite{JAM2012} states that if $\Delta[f]$ satisfies \eqref{intermediate_2}, then, \eqref{intermediate_1} implies the following holds with probability 1:
\[\limsup_{F\rightarrow\infty}\frac{1}{F}\sum_{f=0}^{F-1}V(T[f]-\theta^*)\leq C_0.\]
Divide both sides by $V$ and rearrange the terms gives the desired result.
\end{proof}

\appendices
\section{Proof of optimality for \eqref{value_function_2}}\label{proof_dp}
\begin{proof}
First of all, for any $k\leq0$, since $R_f(P_j)\geq0,~\forall j\in\{1,2,\cdots,|\mathcal{P}|\}$, the best solution is not choosing any PSD allocation at all. This results in $m[k]=0,~\forall k\leq0$. Next, for $k>0$, we use induction to show $m[k]$ indeed satisfies the iteration \eqref{value_function_2}.
\begin{enumerate}
  \item Base case $k=1$. Since $K(\alpha_i,P_j)$ is a positive integer, \eqref{value_function_2} gives $m[1]=\min_jR(P_j)$. On the other hand, the best solution is to place one weight which has the least penalty $R(P_j)$, which obviously implies $m[1]$ satisfies the iteration \eqref{value_function_2}.
  \item Suppose $m[k]$ satisfies \eqref{value_function_2} for all $k\in\{1,2,\cdots,k_0\}$, then, we want to prove $m[k_0+1]$ also satisfies \eqref{value_function_2}. We show that
      \begin{align}\label{case_1}
      m[k_0+1]\leq& \min_{j\in\{1,\cdots,|\mathcal{P}|\}}~R_f(P_j)\nonumber\\
       &+\sum_{i=1}^{|\mathcal{A}|}\phi(\alpha_i)m[k_0+1-K(\alpha_i,P_j)]
      \end{align}
      and
      \begin{align}\label{case_2}
      m[k_0+1]\geq &\min_{j\in\{1,\cdots,|\mathcal{P}|\}}~R_f(P_j) \nonumber\\
      &+\sum_{i=1}^{|\mathcal{A}|}\phi(\alpha_i)m[k_0+1-K(\alpha_i,P_j)]
      \end{align}
      respectively.

      \textit{Proof of \eqref{case_1}:} Consider the following specific allocation policy: Pick any $P_j\in\mathcal{P}$, and then adopt the policy which solves \eqref{appr_aver_1}-\eqref{appr_aver_2} for $k=k_0+1-K(\alpha_i,P_j)$. By induction hypothesis, the resulting penalty of this policy is $R_f(P_j)+\sum_{i=1}^{|\mathcal{A}|}\phi(\alpha_i)m[k_0+1-K(\alpha_i,P_j)]$. Since this is a specific policy, it follows,
      \begin{align*}
      &m[k_0+1]\leq R_f(P_j)+\sum_{i=1}^{|\mathcal{A}|}\phi(\alpha_i)m[k_0+1-K(\alpha_i,P_j)],\\
      &\forall j\in\{1,\cdots,|\mathcal{P}|\},
      \end{align*}
      which implies \eqref{case_1}.
     % \[m[k_0+1]\leq \min_{j\in\{1,\cdots,|\mathcal{P}|\}}~R(P_j)+\sum_{i=1}^{|\mathcal{A}|}\phi(\alpha_i)m[k_0+1-K(\alpha_i,P_j)].\]

      \textit{Proof of \eqref{case_2}:} Suppose $\hat{P}(t_f),~\hat{P}(t_f+1),\cdots,~\hat{P}(t_{f+1}-1)$ is a sequence of PSD allocations generated by the optimal policy minimizing \eqref{appr_aver_1}-\eqref{appr_aver_2} for $k=k_0+1$. Given the event $\alpha(t_f)$ and allocation $\hat{P}(t_f)$ at time slot $t_f$, we look at the remaining sequence of PSD allocations $\hat{P}(t_f+1),\cdots,~\hat{P}(t_{f+1}-1)$. It can be regarded as a sequence generated by a specific policy under the constraint
      \[\sum_{t=t_f+1}^{t_{f+1}-1}K(\alpha(t),\hat{P}(t))\geq k_0+1-K(\alpha(t_f),\hat{P}(t_f)),\]
      which does not necessarily minimize \eqref{appr_aver_1}-\eqref{appr_aver_2} when $k=k_0+1-K(\alpha(t_f),\hat{P}(t_f))$. Formally, this idea can be expressed as
      \begin{align*}
      &m[k_0+1]\\
      =&\expect{\sum_{t=t_f}^{t_{f+1}-1}R_f(\hat{P}(t)):\sum_{t=t_f}^{t_{f+1}-1}K(\alpha(t),\hat{P}(t))\geq k_0+1}\\
      =&\mathbb{E}\left[R_f(\hat{P}(t_f))+\mathbb{E}\left[\sum_{t=t_f+1}^{t_{f+1}-1}R_f(\hat{P}(t)):\right.\right.\\
      &\left.\left.\left.\sum_{t=t_f+1}^{t_{f+1}-1}K(\alpha(t),\hat{P}(t))\geq k_0+1-K(\alpha(t_f),\hat{P}(t_f))\right|\mathcal{F}_{t_f}\right]\right]\\
      \geq&\mathbb{E}\left[R_f(\hat{P}(t_f))+\min_{P(t)\in\mathcal{P}}\left[\sum_{t=t_f+1}^{t_{f+1}-1}R_f(P(t)):\right.\right.\\
      &\left.\left.\left.\sum_{t=t_f+1}^{t_{f+1}-1}K(\alpha(t),P(t))\geq k_0+1-K(\alpha(t_f),\hat{P}(t_f))\right|\mathcal{F}_{t_f}\right]\right]\\
      =&\expect{R_f(\hat{P}(t_f))+m[k_0+1-K(\alpha(t_f),\hat{P}(t_f))]},
      \end{align*}
      where $\mathcal{F}_{t_f}$ denotes all the system information up until time slot $t_f$ and the last equality follows from induction hypothesis. Finally, since $\hat{P}(t_f)$ also comes from a specific policy, it follows,
      \begin{align*}
      m[k_0+1]\geq&\min_{j\in\{1,\cdots,|\mathcal{P}|\}} R_f(P_j)\\
      &+\expect{m[k_0+1-K(\alpha(t_f),P_j)]}\\
      =&\min_{j\in\{1,\cdots,|\mathcal{P}|\}} R_f(P_j)\\
      &+\sum_{i=1}^{|\mathcal{A}|}\phi(\alpha_i)m[k_0+1-K(\alpha_i,P_j)]
      \end{align*}

\end{enumerate}
Overall, we proved the optimality.
\end{proof}

\bibliographystyle{unsrt}
\bibliography{bibliography/refs}

\end{document}